\documentclass[submission,copyright,creativecommons]{eptcs}
\providecommand{\event}{ACT 2024} 

\usepackage{iftex}

\ifpdf
  \usepackage{underscore}         
  \usepackage[T1]{fontenc}        
\else\usepackage{breakurl}           
\fi

\usepackage[utf8]{inputenc}
\usepackage{amsmath,amssymb}
\usepackage{proof}
\usepackage{mathtools}
\usepackage{xfrac}
\usepackage{physics}
\usepackage{tensor}
\usepackage{stmaryrd}
\usepackage{thm-restate}
\usepackage{stmaryrd}
\usepackage{wasysym}
\usepackage{booktabs}
\usepackage{tikz}
\usetikzlibrary{arrows.meta,calc,positioning,decorations.pathmorphing,fit,shapes.misc,patterns,patterns.meta}

\usepackage{apxproof}

\usepackage{graphicx}
\usepackage{caption}
\usepackage{subcaption}
\graphicspath{ {./img/} }
\usetikzlibrary{cd}
\usetikzlibrary{decorations.pathmorphing,positioning,calc}
\tikzset{arrow/.style={-stealth,thick}}
\tikzset{parrow/.style={-stealth,decoration={snake,amplitude=1pt,segment length=6pt,post length=2pt},decorate,thick}}
\tikzset{conf/.style={}}
\tikzset{psplit/.style={draw, circle, inner sep=2pt}}
\tikzset{csplit/.style={fill, circle, inner sep=2pt}}
\tikzset{weight/.style={font=\small,midway}}
\usepackage{centernot}
\usepackage{hyperref}
\usepackage{booktabs}
\usepackage{microtype}
\usepackage{mathtools}
\usepackage{amsthm}
\usepackage{thmtools}
\usepackage{thm-restate}
\usepackage{apxproof}
\usepackage{tikz}
\usepackage{proof}
\usepackage{physics}
\usepackage{stmaryrd}
\usepackage{xfrac}
\usepackage{tensor}
\usepackage{soul}

\newcommand{\qtop}[1][\rho]{\downarrow_{#1}\!\!}
\newcommand{\psum}[1]{\tensor[_{#1}]{\oplus}{}}

\newcommand{\hilbert}{\mathcal{H}}
\newcommand{\hilb}{\hilbert}

\newcommand{\singleton}{\overline}
\newcommand{\distelem}[2]{#1 \cdot #2}

\newcommand{\cfield}{\mathbb{C}}
\DeclareMathOperator{\supp}{supp}

\newcommand{\kp}{\ket{\psi}}
\newcommand{\kf}{\ket{\phi}}
\newcommand{\kz}{\ket{0}}
\newcommand{\ko}{\ket{1}}
\newcommand{\kpl}{\ket{+}}
\newcommand{\km}{\ket{-}}

\newcommand{\dm}[1][]{{D \hspace{-0.03cm} M}_{#1}}

\newcommand{\ef}[1][]{\mathbb{Q}_{#1}}

\newcommand{\blank}{{\,\cdot\,}}
\newcommand{\I}{\mathbb{I}}

\newcommand{\littletaller}{\mathchoice{\vphantom{\big|}}{}{}{}}
\newcommand\restr[2]{{
\left.\kern-\nulldelimiterspace 
#1 
\littletaller 
\right|_{#2} 
}}

\newcommand{\E}{\mathcal{E}}

\newcommand{\iso}{\cong}
\newcommand{\ealg}[1][E]{\mathbb{#1}}
\newcommand{\emon}[1][\ealg]{D_{#1}}
\newcommand{\qmon}[1][C]{Q_{#1}}
\newcommand{\pcmcat}{\mathbf{PCM}}
\newcommand{\eacat}{\mathbf{EA}}
\newcommand{\setcat}{\mathbf{Set}}

\newcommand{\powset}{\mathcal{P}}
\newcommand{\peval}[2]{{#1}\downharpoonleft_{#2}}
\newcommand{\coalgcat}[1][F]{\mathbf{Set}_{#1}}
\newcommand{\PL}[1]{\mathcal{P}({#1})^L}

\newtheorem{desiderata}{Desiderata}

\newtheorem{theorem}{Theorem}
\newtheorem{corollary}{Corollary}
\newtheorem{lemma}{Lemma}

\newtheorem{example}{Example}

\newtheorem{definition}{Definition}

\title{A Coalgebraic Model of Quantum Bisimulation}

\author{Lorenzo Ceragioli
\institute{IMT School for Advanced Studies, Lucca, Italy}
\email{lorenzo.ceragioli@imtlucca.it}
\and 
Elena Di Lavore
\institute{University of Pisa, Pisa, Italy}
\email{elena.dilavore@di.unipi.it}
\and
Giuseppe Lomurno
\institute{University of Pisa, Pisa, Italy}
\email{giuseppe.lomurno@phd.unipi.it}
\and
Gabriele Tedeschi
\institute{University of Pisa, Pisa, Italy}
\email{gabriele.tedeschi@phd.unipi.it}
}

\newcommand{\titlerunning}{A Coalgebraic Model of Quantum Bisimulation}
\newcommand{\authorrunning}{L. Ceragioli, E. Di Lavore, G. Lomurno, G. Tedeschi}

\hypersetup{
  bookmarksnumbered,
  pdftitle    = {\titlerunning},
  pdfauthor   = {\authorrunning},
  pdfsubject  = {EPTCS},               
  pdfkeywords = {coalgebra, quantum computing, quantum effects, bisimilarity, graded monads, effect algebra} 
}

\begin{document}
\maketitle

\begin{abstract}
	Recent works have shown that defining a behavioural equivalence that matches the observational properties of a quantum-capable, concurrent, non-deterministic system is a  surprisingly difficult task.
	We explore coalgebras over distributions taking weights from a generic effect algebra, which subsumes probabilities and \emph{quantum effects}, a physical formalism that represents the probabilistic behaviour of an open quantum system.
	To abide by the properties of quantum theory, we introduce monads graded on a partial commutative monoid, intuitively allowing composition of two processes only if they use different quantum resources, as prescribed by the \emph{no-cloning theorem}.
	We investigate the relation between an open quantum system and its probabilistic counterparts obtained when instantiating the input with a specific quantum state.
	We consider Aczel-Mendler and kernel bisimilarities, advocating for the latter as it characterizes quantum systems that exhibit the same probabilistic behaviour for all input states.
	Finally, we propose operators on quantum effect labelled transition systems, paving the way for a process calculi semantics that is parametric over the quantum input.
\end{abstract}

\section{Introduction}

The recent development of quantum technologies calls for grounded methods for modelling and verifying computing systems that
exploit quantum phenomena like superposition and entanglement.
In particular, concurrent processes are of main interest, like communication protocols~\cite{nurhadiQuantumKeyDistribution2018, bennettQuantumCryptographyPublic2014, longQuantumSecureDirect2007, gaoQuantumPrivateQuery2019} and distributed implementations of algorithms via the \emph{Quantum Internet}~\cite{caleffiQuantumInternetCommunication2018,zhangFutureQuantumCommunications2022}.
Such systems are essentially composed by a set of non-deterministic agents that run in parallels, can communicate, and have access to a quantum memory over which they can perform operations.
The simplest quantum memory is the qubit, which takes values from a complex vector space: it can be in one of two standard basis states $\kz$ and $\ko$, as well as in any linear combination of them.
Agents can observe the state of a qubit only through a measurement on a chosen basis, a special kind of quantum operation
that destroys the qubit and returns a probabilistic classical outcome, which depends on the qubit state.
As an additional constraint, 
the \emph{no-cloning theorem} forbids copying qubits, which thus cannot be shared by multiple agents~\cite{nielsenQuantumComputationQuantum2010}.
Despite the considerable effort~\cite{lalireRelationsQuantumProcesses2006,davidsonFormalVerificationTechniques2012,fengAutomaticVerificationQuantum2015a,  dengBisimulationsProbabilisticQuantum2018,ceragioliQuantumBisimilarityBarbs2024}, there is still no clear notion of behavioural equivalence for such systems, and most of the proposed ones lack a decision procedure.

In~\cite{lics}, the authors have proposed a new semantic model for quantum protocols, namely quantum effect labelled transition systems (qLTSs).
These use \emph{quantum effects} as weights, in the same way probabilities are used in a probabilistic labelled transition system (pLTS).
Coming from quantum measurement theory, quantum effects represent the observable probabilistic properties that a quantum system   expresses when initialized with a given input quantum state.
This indeed allows modelling the agents in quantum protocols: they are essentially non-deterministic procedures acting on an initial quantum state, which is known only when a specific implementation is used in a concrete setting.
Roughly, qLTSs model processes that are parametric over their quantum input.

In this paper we build on this concept, characterizing quantum effects and probabilities as effect algebras.
We investigate their (sub)distributions, and we give effect labelled transition systems ($\ealg$LTSs), a uniform coalgebraic framework encompassing non-deterministic, probabilistic and quantum concurrent systems.
The coalgebraic language is well suited to treat dynamical systems in their essential features, and allows us to extend properties and constructs of probabilistic systems to quantum ones.
We introduce monads graded on a partial commutative monoid (PCM),
allowing us to grade effect distributions over their quantum resources, 
the copy of which is forbidden by the no-cloning theorem.
Transition systems graded over this PCM intuitively represent quantum computations that consume some quantum resources: only computations using disjoint resources
can be composed.
Besides, thanks to our peculiar grading, we define a commutative Kronecker product of effects, and thus a commutative multiplication of effect distributions, generalizing the joint distribution operator of the probabilistic case \cite{bartelsHierarchyProbabilisticSystem2004}.

We investigate kernel and Aczel-Mendler (AM) bisimilarities for $\ealg$LTSs, and relate them by generalizing previous results.
By applying our findings to the quantum setting, we prove that each quantum state $\rho$ defines a functor from qLTSs to pLTSs that “instantiates” a quantum process to the probabilistic behaviour it exhibits when the measured quantum state is $\rho$.
We study which equivalences are preserved and/or reflected by these functors, giving us a categorical way to pinpoint the correct notion of equivalence for quantum systems, i.e. the one that equates all and only the processes that exhibit the same observable probabilistic behaviour.
Both AM and kernel bisimilarities correctly relate indistinguishable quantum processes only.
However, only the latter is complete and relates all the indistinguishable processes, given that the weights of the qLTS are taken from a finite effect algebra of quantum effects.

Finally, we investigate operators over $\ealg$LTSs, paving the way for an $\ealg$LTS semantics of quantum process calculi.
We provide a generalized, compositional parallel operator which can model different notions of synchronization (CCS, CSP, ACT) and different kinds of “weights” (nondeterministic, probabilistic or quantum systems).
In addition, we introduce a purely quantum operator of “partial evaluation” of qLTSs that instantiates the value of some of the qubits.
These operators are defined as functors between $\ealg$LTSs, thus they preserve bisimilarity.
While similar operators in the literature typically act on the final coalgebra, treating them as functors allows for “multi sorted” operators: partial evaluation reduces the resources used by
the computation;
parallel composition joins them by increasing the number of qubits (only if they are compatible, i.e. of different quantum systems).
A qLTS semantics for quantum processes would reduce the problem of deciding probabilistic bisimilarity of two protocols for all possible input quantum states to just checking kernel bisimilarity of two given qLTSs, which can be solved via standard techniques~\cite{jacobsFastCoalgebraicBisimilarity2023a}.
%

\paragraph{Related Works.}
Usually, works from the process algebra literature~\cite{lalireRelationsQuantumProcesses2006,dengOpenBisimulationQuantum2012,davidsonFormalVerificationTechniques2012,dengBisimulationsProbabilisticQuantum2018,ceragioliQuantumBisimilarityBarbs2024,aplas} define the semantics of quantum processes via pLTSs, and strive to adapt probabilistic bisimilarity to capture the peculiar observable properties of quantum values.
The defined pLTSs are made of configurations, i.e.\ pairs of quantum values and syntactic processes, impeding algorithmic verification of processes when the quantum input is not given (the only symbolic approach~\cite{fengSymbolicBisimulationQuantum2014} has been proved too strict~\cite{ceragioliQuantumBisimilarityBarbs2024}).
We instead introduce a purely quantum transition system: we do not represent directly quantum values but only their observable probabilistic features in the form of effects.

This work builds on~\cite{lics}:
 we use the language of distribution monads, their coalgebras and their coalgebraic bisimilarity to generalize the concepts of quantum distribution, quantum transition systems and Larsen-Skou bisimilarity presented there.

Noticeably, our effect distribution monad subsumes both the probability distribution monad, the quantum effect distributions of~\cite{lics} and the quantum monad of~\cite{abramskyQuantumMonadRelational2017}, which is based on projectors, a subset of quantum effects.
Our graded monad definition extends the non-graded effect monad presented in~\cite{jacobsProbabilitiesDistributionMonads2011}, which takes weights in effect monoids, and not in the larger class of graded effect monoids.

Some similar notion to our $\mathbb{E}$LTSs has already been proposed: effect-valued Quantum Markov Chain of~\cite{gudderQuantumMarkovChains2008} uses sequential effect composition instead of tensor product; QLTS of~\cite{ogawaCoalgebraicApproachEquivalences2014} uses superoperators instead of effects, so to capture also non-destructive measurements.
The author of~\cite{ogawaCoalgebraicApproachEquivalences2014} introduces also two different notions of bisimilarity, that we recover in our minimal, effect-based setting as AM and kernel bisimilarity.
However, none of these works feature nondeterminism, which are needed for modelling protocols, nor they investigate correctness by comparing their findings with the probabilistic behaviour of quantum systems.
A final novelty of our work is investigating operators over labelled transition systems, suitable for modelling quantum protocols, e.g. via process calculi.

%

\paragraph{Synopsis.}
In~\autoref{sec:bg} we give some background about effect algebras, quantum computing and quantum effects. 
In~\autoref{sec:efDist} we present effect distributions and their monads.
In~\autoref{sec:elts} we investigate coalgebras on effect distributions 
and their bisimilarity, and we discuss $\ealg$LTSs and their operators in~\autoref{sec:operators}.  
Finally, we conclude in~\autoref{sec:conclude}, and we give the proof sketches in the appendix.

\section{Background}\label{sec:bg}

We recall the definitions of partial commutative monoid and effect algebra.
Then, we give some background on quantum computing and quantum effects.
We refer to~\cite{nielsenQuantumComputationQuantum2010} for further reading on quantum computing, to~\cite{heinosaariMathematicalLanguageQuantum2011} for the 
details of quantum effects, and to~\cite{jacobsConvexityDualityEffects2010} for their algebraic treatment.

\subsection{Partial Commutative Monoid and Effect Algebra}
Partial commutative monoids obey the properties of commutative monoids, but $+$ is not always defined.
\begin{definition}
	A partial commutative monoid (PCM) is a tuple $\langle M, 0, +\rangle$ (often referred as $M$) with $0 \in M$ and $+ : M \times M \to M$ a partial binary operation on $M$ such that for all $a, b, c \in M$ the following hold:
	\begin{itemize}
		\item (Commutativity) $a \perp b$ implies $b \perp a$ and $a + b = b + a$;
		\item (Associativity) $b \perp c$ and $a \perp (b + c)$ implies $a \perp b$ and $(a + b) \perp c$ and also $(a + b) + c = a + (b + c)$;
		\item (Zero) $0 \perp a$ and $0 + a = a$.
	\end{itemize}
	Here, $a \perp b$ means that $a$ and $b$ are orthogonal, i.e.\ $a + b$ is defined.
	A PCM homomorphism is a function  $f : M \to N$ on the underlying carrier sets such that $f(0) = 0$, and $a \perp b$ implies $f(a) \perp f(b)$ and $f(a + b) = f(a) + f(b)$.
	PCMs and their homomorphisms form the category $\pcmcat$.
\end{definition}
Every PCM induces a preorder on the carrier set $M$ where $a \preceq b$ if and only if $\exists c \in M \ldotp a + c = b$.

Effect algebras are a special kind of PCMs for which an inverse operation is defined.
\begin{definition}
	An effect algebra~\cite{jacobsConvexityDualityEffects2010} is a tuple $\langle \ealg, 0, +, \blank' \rangle$ (often referred as $\ealg$) with $\langle \ealg, 0, +\rangle$ a PCM and $\blank' : E \rightarrow E$ a unary operation such that,
	for all $e \in E$:
	\begin{itemize}
		\item $e' \in E$ is the unique element in $E$ such that $e + e' = 1$ with $1 = 0'$;
		\item $e \perp 1$ implies $e = 0$.
	\end{itemize}
	An effect homomorphism is a PCM homomorphism that also preserves $\blank'$.
	The category $\eacat$ of effect algebras and effect homomorphisms is the full subcategory of $\pcmcat$ whose objects are effect algebras.
\end{definition}
Effect algebras have a partial order $\sqsubseteq$ (defined as $\preceq$) and a partial operation $e_1 - e_2$ returning the unique $e_3$ such that $e_2 + e_3 = e_1$.
%
$\eacat$ is a symmetric monoidal category with $2 = \{0,1\}$ its unit object.
There is a bijective correspondence between morphisms $\ealg_A \otimes \ealg_B \to \ealg_C$ and bihomomorphisms $\ealg_A \times \ealg_B \to \ealg_C$
(a bihomomorphism is such that the morphisms obtained by fixing either object are homomorphisms).


An example of effect algebras are probabilities: real numbers in the interval $[0, 1]$ with $+$ the arithmetic sum and $e'$ defined as $1 - e$.
Probabilities allow for defining \emph{probability (sub)distributions} over a given set $S$, i.e. functions $\Delta : S \to
	[0,1]$ such that $\sum_{s \in S}\Delta(s) \leq 1$.
%
For each $s \in S$, we let $\singleton{s}$ be the \emph{point distribution} that assigns $1$ to $s$.
Given a finite set of non-negatives reals $\{p_i\}_{i \in I}$ such that $\sum_{i \in I} p_i \leq 1$, the weighted sum $\sum_{i \in I} \distelem{p_i}{\Delta_i}$ defines a distribution such that $(\sum_{i \in I} \distelem{p_i}{\Delta_i})(s) = \sum_{i \in
		I}p_i\Delta_i(s)$.

Probability distributions form a \emph{convex set}~\cite{bonchiPowerConvexAlgebras2017}, meaning that for all distributions $\Delta, \Theta$ and for all real $p \in [0,1]$ there exists a distribution  $\Delta \psum{p} \Theta$ defined as $\distelem{p}{\Delta_1} + \distelem{(1 - p)}{\Delta_2}$.
Given a function $f$ between convex sets $X$ and $Y$, we call $f$ convex if it preserves the $\psum{p}$ operator, i.e.\ if $f(x_1 \psum{p} x_2) = f(x_1) \psum{p} f(x_2)$.
We denote as $\mathbf{Conv}(X, Y)$ the set of convex functions between $X$ and $Y$.




\subsection{Quantum Computing}



A (finite-dimensional) \emph{Hilbert space}, denoted as $\hilbert$, is a
complex vector space equipped with a binary operator $\braket\blank: \hilbert
	\times \hilbert \rightarrow \cfield$ called \emph{inner product}, defined as
$\braket{\psi}{\phi} = \sum_i \alpha_i^*\beta_i$, where $\kp =
	(\alpha_1,\ldots,\alpha_i)^T$ and $\kf = (\beta_1,\ldots,\beta_i)^T$, with $T$ the transpose.
We indicate column vectors as $\kp$ and their conjugate transpose as
$\bra\psi = \kp^\dagger$.
The state of an isolated physical system is represented as a \emph{unit vector}
$\kp$ (called \emph{state vector}), i.e.\ a vector such that $\braket\psi =
	1$.
The two-dimensional Hilbert space $\cfield^2$ is called a \emph{qubit}.
The
vectors $\{\ket0 = (1,0)^T, \ket1 = (0,1)^T\}$ form an orthonormal basis of
$\cfield^2$, called the \emph{computational basis}. Other important vectors in
$\cfield^2$ are $\ket+ = \frac{1}{\sqrt{2}}(\ket0 + \ket1)$ and $\ket-
	= \frac{1}{\sqrt{2}}(\ket0 - \ket1)$, which form the \emph{Hadamard basis}.
Intuitively, different bases represent different observable properties of a quantum system.
Note that $\kpl$ and $\km$ are non-trivial linear combinations of $\kz$ and $\ko$, roughly meaning that
the property associated with the computational basis is undetermined in $\kpl$ and $\km$.
In the quantum jargon, $\kpl$ and $\km$ are \emph{superpositions} with respect to the computational basis.
Symmetrically, $\kz$ and $\ko$ are superpositions with respect to the Hadamard one.

We represent the state space of a composite physical system as the \emph{tensor
	product} of the state spaces of its components. Let $\hilbert$ and
$\hilbert'$ be $n$ and $m$-dimensional Hilbert spaces: their tensor product
$\hilbert \otimes \hilbert'$ is an $n\cdot m$ Hilbert space. Moreover, if
$\{\ket{\psi_1}, \ldots, \ket{\psi_n}\}$ and $\{\ket{\phi_1}, \ldots,
	\ket{\phi_m}\}$ are bases of respectively $\hilbert$ and $\hilbert'$, then $
	\{\ket{\psi_i}\otimes\ket{\phi_j} \mid i = 1, \ldots, n, j = 1, \ldots, m\}$ is
a basis of $\hilbert \otimes \hilbert'$, where $\kp \otimes \kf$
is the Kronecker product.
	We often omit the tensor product and write $\kp\kf$ or $\ket{\psi\phi}$.
	Note that such product is not commutative.
	Categorically, finite-dimensional Hilbert spaces with the Kronecker product and the conjugate transpose form the dagger compact category $\mathit{FDHilb}$.
	Further references are available in \cite{coeckePicturingQuantumProcesses2017}.




The density operator formalism puts together quantum systems and probability by considering
mixed states, i.e.\ \emph{probability distributions of quantum states}.
A point distribution $\singleton{\kp}$ (called a pure state) is represented by the matrix $\ketbra{\psi}$.
In general, a probability distribution $\Delta$ of $n$-dimensional states is represented as the matrix $\rho \in \cfield^{n\times n}$, known as its \emph{density operator}, with $\rho = \sum_i \Delta(\psi_i)\ketbra{\psi_i}$.
%
Recall that a complex matrix $N$ is called \emph{positive semi-definite}, shortly positive, when $\ev{N}{\psi} \geq 0$ for all $\kp$.
The \emph{L\"owner order} is the partial order defined by $L \sqsubseteq L'$ whenever $L' - L$ is positive.
Given a $d$-dimensional Hilbert space $\hilbert$, density operators coincide with the positive matrices in $\mathbb{C}^{d\times d}$ of trace one, we denote them as $\dm[\hilbert] = \big\{\,\rho \in \mathbb{C}^{d\times d} \mid \rho \sqsupseteq 0_d, \ \tr(\rho) = 1\,\big\}$.
Density operators form a convex set, where the convex combination operator is defined by $\rho \psum{p} \sigma = p\rho + (1-p)\sigma$.

Density operators can describe the state of a subsystem of a
composite quantum system.
Let $\hilbert_S$ denote the Hilbert space of a physical system $S$, then $\hilbert_{S_1} \otimes \hilbert_{S_2}$
is the Hilbert space of a composite system with subsystems $S_1$ and $S_2$. 
Given a (not necessarily separable) $\rho \in \hilbert_{S_1} \otimes \hilbert_{S_2}$, the \emph{reduced density operator} of system $S_1$, $\rho_1 = \tr_{S_2}(\rho)$, describes the state of $S_1$, with $\tr_{S_2}$ the \emph{partial trace over $S_2$}, defined as
the linear transformation such that
$\tr_{S_2}(\ketbra{\psi}{\psi'} \otimes \ketbra{\phi}{\phi'}) =
\ketbra{\psi}{\psi'}\tr(\ketbra{\phi}{\phi'})$ for each $\ketbra{\psi}{\psi'} \in \dm[\hilbert_{S_1}]$ and $\ketbra{\phi}{\phi'} \in \dm[\hilbert_{S_2}]$.

\subsection{Quantum Effects}

\emph{Quantum measurements} are needed for describing systems that exchange
information with the environment.
Performing a measurement on a quantum state returns a probabilistic classical result and either destroys or otherwise changes the quantum system. We focus in this paper on destructive measurements.

The simplest kind of measurements are \emph{quantum effects} (simply called effects in quantum textbooks~\cite{heinosaariMathematicalLanguageQuantum2011}), i.e.\ yes-no tests over quantum systems.
%
Each effect can be represented as a positive matrix smaller than the identity in the L\"owner order.
We denote the set of effects on a $d$-dimensional Hilbert space $\hilbert$ as $\ef[\hilbert] = \big\{\,L \in \mathbb{C}^{d\times d} \mid 0_d \sqsubseteq L \sqsubseteq \I_d\,\big\}$, where $\I_d$ is the $d \times d$ identity matrix.
The probability of getting a ``yes'' outcome when measuring an effect $L$ on a state $\kp$ is given by the \emph{Born rule} $tr(L\rho)$.
Effects of dimension $d$ form an effect algebra with the matrix sum, $\I_d$ as $1$ and $L' = \I_d - L$. Furthermore, the induced partial order $\sqsubseteq$ is exactly the L\"owner order.

A measurement with $n$ outcomes is a set $\{L_1, \ldots L_n\}$ of effects,
such that the \emph{completeness} equation $\sum_{i=1}^n
	L_i = \I$ holds. If the state of the system is $\rho$ before the
measurement, then the probability of the $i$ outcome occurring is $p_i = tr(L_i\rho)$.
%
As examples of measurements, consider $M_{01}$ and $M_{\pm}$ that project a state into the elements of the computational and Hadamard basis of $\cfield^2$ respectively,
with $M_{01}$ defined as $\{\ketbra{0}, \ketbra{1}\}$ and $M_\pm$ as $\{\ketbra{+}, \ketbra{-}\}$.
%
Applying the measurement $M_{01}$ on $\kz$ returns the outcome associated with $\ketbra{0}$ with probability $1$.
When measuring $\kpl$, instead, the same result occurs with probability $\frac{1}{2}$.
Notice that a measurement for a composite system may measure only some of the qubits, e.g. $\{ \ketbra{0} \otimes \I, \ketbra{1} \otimes \I\}$ measures (in the computational basis) the first qubit of a pair.

Density operators and effects are dual, as effects are isomorphic to the convex functions from the set of density operators to the probability interval.
The isomorphism is given by the Born rule.
\begin{theorem}\label{thm:effectsiso}
	It holds that $\ef[\hilbert] \iso \mathbf{Conv}(\dm[\hilbert], [0,1])$ through the isomorphism $L \mapsto \lambda \rho\ldotp tr(L\rho)$~\cite{heinosaariMathematicalLanguageQuantum2011}.
\end{theorem}
Roughly, effects can be considered as probabilities \emph{parametrized} on an unknown quantum state.

\section{Effect Distributions}\label{sec:efDist}

Here we lay some fundamental definitions about effect distributions, which generalize probability distributions by using elements of a generic effect algebra $\ealg$ as weights, instead of the usual $[0,1]$ interval.
We show that the features of effect distributions descend from the algebraic structure of their weights: if the latter have a (commutative) multiplication, then effect distribution form a (commutative) monad; if we have a morphism of weights, we get a natural transformations of distributions.
We exemplify these concepts on probability distributions and on quantum effect distributions.
Finally, we show that, like probability distributions, quantum distributions do form a commutative monad, albeit a graded one.

Following the work of~\cite{jacobsConvexityDualityEffects2010}, for each effect algebra $\ealg$ we build a functor of effect distributions.

\begin{definition}[Effect Distributions]\label{def:efFunct}
	Given an effect algebra $\langle \ealg, 0, +, \blank'\rangle$ we define the functor of \textit{effect (sub)distributions} $D_{\ealg}: \mathbf{Set} \to \mathbf{Set}$ by
	\[
		D_{\ealg}X = \left\{\,
		\Delta \in {\ealg}^{X}
		\;\middle| \;
		\supp(\Delta) \text{ is finite,} \,
		\smashoperator{\sum_{x \in \supp(\Delta)}} \Delta(x) \sqsubseteq 1_{\ealg} \right\}
		\qquad
		(\emon f)(\Delta) = \lambda y \in Y.\sum_{\mathclap{{x \in f^{-1}(y)}}} \Delta(x)
	\]
	where supp$(\Delta)$ is the set $\{\,x \in X\;|\;\Delta(x) \neq 0\,\}$, and $\sum$ is the n-ary sum in $\ealg$.
\end{definition}

Our running examples will be sub-probability distributions $D_{[0,1]}$, featuring weights the unit interval, and quantum distributions over some Hilbert space $\hilbert$, associated to the effect algebra $\ef[\hilbert]$.
Notice how, if $\hilbert$ is $1$ dimensional, the effect algebra $\ef[\hilbert]$ is isomorphic to $[0,1]$, and $D_{\ef[\hilb]}$ to $D_{[0,1]}$.
We will often write an effect distribution $\Delta = \{x_i \mapsto e_i\}_{x_i \in X}$ in a compact form, as $\Delta = \sum_{x_i \in X} e_i \bullet x_i$.
Thus, some examples of distributions on $X = \{a, b\}$ could be
\[
	\frac{1}{2} \bullet s + \frac{1}{2}\bullet t \in D_{[0, 1]}X \iso D_{\ef[\mathbb{C}]}X \qquad \ketbra{0} \bullet s + \ketbra{1} \bullet t \in D_{\ef[\mathbb{C}^2].}
\]
Intuitively, the distribution on the left model a random choice, where we toss a fair coin to decide between $s$ and $t$. The one on the right is a quantum choice, i.e. a measurement, 
where we choose by measuring a qubit in the computational basis and by taking $s$ or $t$ according to the result.

Transformation of weights gives us transformations of distributions. 
Formally, we have the following theorem, rephrasing Proposition 21 of~\cite{jacobsConvexityDualityEffects2010}.
\begin{restatable}{theorem}{morphisnattrans}\label{thm:morphisnattrans}
	Each effect morphism $m: \ealg[E] \to \ealg[F]$ yields a natural transformation $m \circ \blank :\emon[{\ealg[E]}] \Rightarrow \emon[{\ealg[F]}]$.
	Moreover, if $m$ is injective then the components of $m \circ \blank$  are injective.
\end{restatable}

For example, for each $\rho \in \dm[\hilbert]$, there is an effect algebra homomorphism $m_\rho$ from $\ef[\hilbert]$ to $[0,1]$ given by $m_\rho(L) = tr(L\rho)$.
Intuitively, this homomorphism applies an input density operator $\rho$ that transforms each effect $L$ to the probability of observing the "yes" outcome when performing the $L$ measurement a quantum system in state $\rho$.
From this we get a transformation from $\Delta \in \emon[{\ef[\hilb]}]X$ to $m_\rho \circ \Delta \in \emon[{[0,1]}]X$.
Different $\rho$ yields different transformations: given $\Delta = \ketbra{0} \bullet s + \ketbra{1} \bullet t$, we have 
$$m_{\ketbra{+}} \circ \Delta = \frac{1}{2} \bullet s + \frac{1}{2} \bullet t \qquad m_{\ketbra{1}} \circ \Delta = 1 \bullet t.$$

It is well known that, when weights have a multiplicative structure, weighted distributions form a monad~\cite{jacobsProbabilitiesDistributionMonads2011,miculanSemiringWeak,bonchiSemiringUpto}.
While probabilities do have a multiplication, forming a monoid in $\eacat$ (called an effect monoid in~\cite{jacobsProbabilitiesDistributionMonads2011}), quantum effects do not.

Notice how two effects 
$L_1, L_2 \in \ef[{\hilb}]$
can be multiplied via the tensor product, but the resulting effect $L_1 \otimes L_2$ is in 
$\ef[{\hilb \otimes \hilb}]$.
Essentially, the carrier of 
$\ef[\hilb]$
is not closed for the tensor operation.
We thus introduce \emph{graded effect monoids}, allowing quantum effects to be tensored together, using grades to keep track of the resulting Hilbert space.

Instead of grading on a monoid or a monoidal category, we use a partial commutative monoid (PCM) as grades.
As we will see, this is useful for ensuring linearity of quantum resources. 
\begin{definition}[Graded Effect Monoid]
	Given a partial commutative monoid $\langle M, 0, +\rangle$, an $M$-graded effect monoid is a graded monoid object in $\eacat$, that is
	\begin{itemize}
		\item for each $m \in M$, an effect algebra $\ealg_m$
		\item for each $m \bot n \in M$, an effect morphism $\nabla_{m,n}: \ealg_m \otimes \ealg_n \to \ealg_{m + n}$
		\item an effect morphism $\eta: 2 \to \ealg_{0}$
	\end{itemize}
	such that the two following diagrams on the left commute
	\[
		\begin{tikzcd}
			(\ealg_m \otimes \ealg_n) \otimes \ealg_o \arrow[r, "\alpha"] \arrow[d, "{\nabla_{m,n} \otimes 1}"'] & \ealg_m \otimes (\ealg_n \otimes \ealg_o) \arrow[d, "{1 \otimes \nabla_{n,o}}"] & 2 \otimes \ealg_m \arrow[d, "\lambda"'] \arrow[r, "\eta \otimes id"] & \ealg_0 \otimes \ealg_m \arrow[ld, "{\nabla_{0, m}}"]  & \ealg_m \otimes \ealg_n \arrow[r, "{B_{\ealg_m, \ealg_n}}"] \arrow[d, "{\nabla_{m,n}}"'] & \ealg_n \otimes \ealg_m \arrow[ld, "{\nabla_{n,m}}"] \\
			\ealg_{m+n} \otimes \ealg_o \arrow[d, "{\nabla_{m+n, o}}"']                                      & \ealg_m \otimes \ealg_{n+o} \arrow[ld, "{\nabla_{m, n+o}}"]                 & \ealg_m                                                              &                                                & \ealg_{m+n}                                                                  &                                              \\
			\ealg_{m+n+o}                                                                                &                                                                     & \ealg_m \otimes 2 \arrow[u, "\rho"] \arrow[r, "id \otimes \eta"']    & \ealg_m \otimes \ealg_0 \arrow[lu, "{\nabla_{m, 0}}"'] &                                                                          &
		\end{tikzcd}
	\]

	If also the diagram on the right commutes, where $B$ is the braiding natural transformation, then we say that the graded effect monoid is \emph{commutative}.

\end{definition}

Each monoid object can be considered as being graded on the degenerate $1$-element monoid $\{0\}$, thus all effect monoids described in \cite{jacobsProbabilitiesDistributionMonads2011} are an example of trivially graded effect monoids.

For the quantum case, we grade our 
effect monoid over the PCM of disjoint sets of quantum systems.
Note that it is indeed impossible to perform multiple measurements on the same quantum state,
because the no-cloning theorem forbids copying quantum data, and 
measurements 
destroy the measured system.
For complying with this physical constraint, our grading keeps track of the qubits measured by the effects, and the partial sum reflects that the composition of effects is defined over disjoint systems only.

\begin{definition}[PCM of Quantum Systems]
	Assume a finite set of quantum systems $Sys = \{ S_i \}$, each associated with the Hilbert space $\hilb_{S_i}$ of finite dimension $d_i$.
	Let $\mathcal{S} = \langle \mathcal{P}(Sys), \emptyset, \uplus\rangle$ be the PCM where $\mathcal{P}(Sys)$ is the powerset of $Sys$ and $\uplus$ is the partial disjoint union, i.e. $C_i \uplus C_j$ is defined only if $C_i \cap C_j = \emptyset$, and in that case $C_i \uplus C_j = C_i \cup C_j$. We associate each collection of systems $C \in \mathcal{P}(Sys)$ with a Hilbert space defined as $\hilb_{\emptyset} = \mathbb{C}$ and $\hilb_{C} =
	\bigotimes_{S \in C} \hilb_{S}$, where we impose that the arguments of the (non-commutative) Kronecker product are ordered according to their indices.
\end{definition}
%
Assume for example $Sys$ to be the finite set of qubits $\{q_1, \ldots, q_n\}$.
Then, our collections $C \in \mathcal{S}$ are subsets of these qubits, and the grading allows us to keep track of which of them are consumed by each measurement, forbidding the tensoring of effects that use a common qubit.
Indeed, the PCM above grades an effect monoid of quantum effects.
\begin{restatable}[Effect Monoid of Quantum Effect]{theorem}{efMonQuEf}\label{thm:efMonQuEf}
	Quantum effects carry a commutative $\mathcal{S}$-graded effect monoid structure, given by:
	\begin{itemize}
		\item The effect algebra $\ealg_C$ is $\ef[\hilb_C]$ for all collection of systems $C \in \mathcal{P}(Sys)$

		\item The operator $\nabla_{C,D}:\ealg_C \otimes \ealg_D \to \ealg_{C \uplus D} $ is denoted as $\boxtimes$ and defined by $L_1 \boxtimes L_2 = Sort_{C, D}(L_1 \otimes_k L_2)$, where $\otimes_k$ is the Kronecker product between $L_1 \in \ef[\hilb_C]$ and $L_2 \in \ef[\hilb_D]$, and $Sort_{C, D}$ is the unitary transformation which ``sorts'' the Hilbert Space $\hilb_C \otimes \hilb_D$ into $\hilb_{C \uplus D}$.
		\item The effect morphism $\eta: 2 \to \ealg_\emptyset$ defined by $\{0 \mapsto 0, 1 \mapsto 1\}$.
	\end{itemize}
\end{restatable}

Since quantum effects are only a PCM-graded monoid and not a monoid in $\eacat$, the resulting distribution functor is a PCM-graded monad and not a monad.
\begin{definition}[Graded Monad]
	Given a partial commutative monoid $\langle M, 0, +\rangle$, an $M$-graded monad on $\setcat$ is a graded monoid in the category of $\setcat$-endofunctors, that is
	\begin{itemize}
		\item for each $m \in M$, an endofunctor $T_m: \setcat \to \setcat$
		\item a natural transformation $\eta: Id \to T_{0}$ called \emph{unit}
		\item for each $m \bot n \in M$, a natural transformation $\mu_{m,n}: T_mT_n \to T_{m + n}$ called \emph{multiplication}
	\end{itemize}
	such that the following diagrams commute
	\[
		\begin{tikzcd}
			T_mT_nT_o \arrow[d, "{\mu_{m,n}T_o}"'] \arrow[r, "{T_m\mu_{n,o}}"] & T_{m}T_{n+o} \arrow[d, "{\mu_{m,n+o}}"] &  & T_m \arrow[rd, Rightarrow, no head] \arrow[d, "\eta T_m"'] \arrow[r, "T_m\eta"] & T_mT_0 \arrow[d, "{\mu_{m,0}}"] \\
			T_{m+n}T_o \arrow[r, "{\mu_{m+n, o}}"']                            & T_{m+n+o}                               &  & T_0T_m \arrow[r, "{\mu_{0,m}}"']                                                & T_m
		\end{tikzcd}
	\]
\end{definition}
PCM-graded monads generalize the notion of monoid-graded monad of~\cite{katsumata2014effect}, asking for an associative multiplication only when the grades are orthogonal.

Es expected, whenever $\{\ealg_m\}$ is a graded effect monoid, $\{\emon[\ealg_m]\}$ has a graded monad structure.
\begin{restatable}[Graded Monads of Graded Effect Monoids]{theorem}{gmofgem}\label{thm:gmofgem}
	If $\{\ealg_m\}$ is an $M$-graded effect monoid, there is a graded monad $\{\emon[\ealg_m]\}$ with unit $\eta: Id \to \emon[\ealg_0]$ and multiplication $	\mu_{m,n}: \emon[\ealg_m]\emon[\ealg_n] \to \emon[\ealg_{m+n}]$ given by
	\[ \eta(x) = 1_{\ealg_0} \bullet x \qquad
		\mu_{m,n}(\sum\nolimits_i e_i \bullet \Delta_i) x = \sum\nolimits_i \nabla_{m,n}(e_i, \Delta_i(x))
	\]
\end{restatable}

For brevity, we will write $Q_C$ for $\emon[{\ef[\hilbert_C]}]$ in the following.
Recall that, in the probabilistic case, $\mu$ corresponds to the weighted sum of probability distributions~\cite{hennessyExploringProbabilisticBisimulations2012}.
For quantum distributions, instead, multiplication models concatenated quantum measurements, by taking the (sorted) Kronecker product of the effects:
given a set of quantum effect distributions $\Delta_i$ in 
$Q_C X$
and a quantum distribution $\Theta$ in 
$Q_{C'} (Q_{C} X)$
associating each distribution $\Delta_i$ with a quantum effect, the multiplication returns a distribution in 
$Q_{C' \uplus C} X$
associating
each $x \in X$ with $\sum_{i} \Theta(\Delta_i) \boxtimes \Delta_i(x)$ (if $C$ and $C'$ are disjoint).
This coincides with the intuition of measuring first the qubits in $C$, and then, based on the outcome, performing a second measurement over $C'$, i.e. $\Delta_i$.

For later use, we finally define commutative graded monads.
Commutativity allows us to reduce the pairing of effect distributions to a distribution of pairs, permitting a well-behaved definition of the parallel composition of effect distributions.
For the (trivially graded) probabilistic case, this corresponds to the joint probability distribution construction à la~\cite{Sokolova2004}.
\begin{definition}[Commutative Graded Monad]\label{def:cgm}
	An $M$-graded monad $\{T_m\}$ on $\setcat$ is strong if it has \emph{left strength} $\sigma_{m, X, Y}: X \times T_mY \to T_m(X \times Y)$ and a \emph{right strength} $\tau_{m, X, Y}: T_mX \times Y \to T_m(X \times Y)$ which respect the monoidal structure $\times$ on $\setcat$ and the graded multiplication and unit of $T$ (the diagrams it must satisfy are just the graded version of the usual ones for strong monads).
	A strong graded monad is \emph{commutative} if for all $m \bot n \in M$ there is a canonical natural transformation $\alpha : T_m X \times T_n Y \to T_{m+n}(X \times Y)$
	defined by any of the two compositions $\mu_{m,n} \circ T_m\sigma_n \circ \tau_m = \mu_{n,m} \circ T_n\tau_m \circ \sigma_n$.
\end{definition}

\begin{restatable}{theorem}{commute}\label{thm:commute}
	If $\{\ealg_m\}$ is a commutative graded effect monoid, then $\{\emon[\ealg_m]\}$ is a commutative graded monad with its canonical strength.
\end{restatable}

\section{$D_{\ealg}$-Coalgebras for Quantum Systems}\label{sec:elts}
%

We first investigate coalgebras defined on effect distributions and their bisimilarities.
Then, we focus on the specific case of quantum systems, discussing the relation between probabilistic and quantum coalgebras and characterizing correctness for a quantum behavioural equivalence.

\subsection{$D_{\ealg}$-Coalgebras and their Bisimilarities}

We start by recalling the definitions of coalgebra and coalgebra homomorphism.
\begin{definition}[Coalgebra]
	Let $F : \setcat \to \setcat$ be an endofunctor on the $\setcat$ category.
	An \emph{$F$-coalgebra} is a pair $(X, c)$, with $X$ an object of $\setcat$, and a morphism $c : X \to FX$ (also written $X \xrightarrow{c} FX$).

	Given two $F$-coalgebras $(X, c)$ and $(Y, d)$,
		A morphism $f : X \to Y$ is an \emph{$F$-coalgebra homomorphism}, written $f : (X, c) \to (Y, d)$, if $Ff \circ c = d \circ f$.
		$F$-coalgebra homomorphisms include the identity and are closed for composition, and thus $F$-coalgebras constitute a category, denoted as $\coalgcat$.
	%
\end{definition}
We consider the two most common kinds of bisimilarity on coalgebras: AM-bisimilarity and kernel bisimilarity, of which we recall the definitions below.

\begin{definition}[Aczel-Mendler Bisimilarity]
	A relation $R \subseteq X \times Y$ is an \emph{AM}-bisimulation  between the $F$-coalgebras $X \xrightarrow{c} FX$ and $Y \xrightarrow{d} FY$ if there exists an $F$-coalgebra $R \xrightarrow{e} FR$ such that the projections $\pi_1: R \to X$ and $\pi_2: R \to Y$ are coalgebra homomorphisms, i.e. the diagram in \autoref{fig:ambisimcd} commutes.
	Two states $x \in X$ and $y \in Y$ are \emph{AM-bisimilar}, written $x \sim_{AM} y$, if $x\,R\,y$ for some AM-bisimulation $R$.
\end{definition}

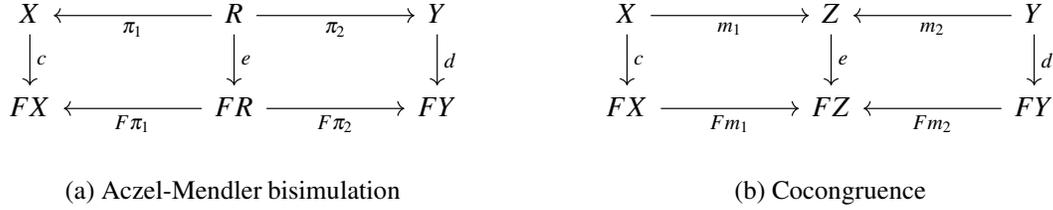
\begin{figure}
	\begin{subfigure}{0.49\textwidth}
		\begin{center}
			\begin{tikzcd}
				X \arrow[d, "c"] &  & R \arrow[d, "e"] \arrow[ll, "\pi_1"] \arrow[rr, "\pi_2"'] &  & Y \arrow[d, "d"] \\
				FX &  &  FR \arrow[ll, "F\pi_1"] \arrow[rr, "F\pi_2"'] &  & FY 
			\end{tikzcd}
		\end{center}
		\caption{Aczel-Mendler bisimulation}
		\label{fig:ambisimcd}
	\end{subfigure}
	\begin{subfigure}{0.49\textwidth}
		\begin{center}
			\begin{tikzcd}
				X \arrow[d, "c"] \arrow[rr, "m_1"'] &  & Z \arrow[d, "e"] &  & Y \arrow[d, "d"] \arrow[ll, "m_2"] \\
				FX \arrow[rr, "Fm_1"'] &  & FZ &  & FY \arrow[ll, "Fm_2"]                                  
			\end{tikzcd}
		\end{center}
		\caption{Cocongruence}
		\label{fig:kbisimcd}
	\end{subfigure}
	\caption{The two commutative diagrams defining AM and Kernel bisimulations, respectively}
\end{figure}

\begin{definition}[Kernel Bisimilarity]
	A \emph{cocongruence} between two $F$-coalgebras $X \xrightarrow{c} FX$ and $Y \xrightarrow{d} FY$ is a cospan $X \xrightarrow{m_1}Z \xleftarrow{m_2} Y$ such that there exists an $F$-coalgebra $Z \xrightarrow{e} FZ$ making $m_1$ and $m_2$ coalgebra homomorphisms, i.e. the diagram in \autoref{fig:kbisimcd} commutes.
	We call \emph{kernel bisimulation} the pullback in $\setcat$ of a cocongruence $X \xrightarrow{m_1}Z \xleftarrow{m_2} Y$, that is the relation $R$ such that $x\,R\,y$ if and only if $m_1(x) = m_2(y)$.
	\emph{Kernel bisimilarity} $\sim_{k}$ is the largest kernel bisimulation.
\end{definition}

In particular, we are interested in the case where the endofunctor $F$ is $D_{\ealg}$ for some effect algebra $\ealg$.
In this case, the equality of kernel and AM bisimilarities entirely depends on the underlying effect algebra $\ealg$.
In line with the results in \cite{GUMM2001185,ogawaCoalgebraicApproachEquivalences2014}, the two relations coincide if and only if $\ealg$ is decomposable.
\begin{definition}[Decomposable Effect Algebra]
	We say that an effect algebra $\ealg$ is \textit{decomposable} if for all $a, b, c, d \in \mathbb{E}$ such that $a \perp b$, $c \perp d$ and $a + b = c + d$, there exists $e_{11}, e_{12}, e_{21}, e_{22} \in \mathbb{E}$ such that
	$a = e_{11}+e_{12}$, $b = e_{21}+e_{22}$, $c = e_{11}+e_{21}$ and $d = e_{12}+e_{22}$.
\end{definition}

\begin{restatable}{theorem}{amkequality}\label{thm:amkequality}
	Let $\ealg$ be an effect algebra.
	Let $X \xrightarrow{c} \emon X$ and $Y \xrightarrow{d} \emon Y$ be two $\emon$-coalgebras.
	For all $x \in X$ and $y \in Y$: $(i)$ $x \sim_{AM} y \Longrightarrow x \sim_{k} y$, and $(ii)$ $x \sim_{k} y \Longrightarrow x \sim_{AM} y$ if and only if $\mathbb{E}$ is decomposable.
\end{restatable}

Another topic we are interested in is the relation between $D_{\ealg}$-coalgebras over possibly different effect algebras, namely $D_{\ealg}$-coalgebras and $D_{\ealg'}$-coalgebras for a possibly different $\ealg'$.
In particular, we want to know under which conditions a mapping from $D_{\ealg}$-coalgebras to $D_{\ealg'}$-coalgebras preserves and reflects bisimilarities.
Note that natural transformations between the functors that define a coalgebra have the important property of preserving bisimilarities, as demonstrated in \cite{bartelsHierarchyProbabilisticSystem2004}.
\begin{theorem}\label{thm:nattransisfunc}
	Let $\alpha : F \Rightarrow G$ be a natural transformation between two functors $F$ and $G : \setcat \to \setcat$.
	The natural transformation $\alpha$ induces a functor, denoted $\alpha \circ \blank : \coalgcat[F] \to \coalgcat[G]$, which maps objects $(X, c)$ to $(X, \alpha_X \circ c)$ and homomorphisms $f : (X, c) \to (Y, d)$ to homomorphisms $f : (X, \alpha_X \circ c) \to (Y, \alpha_Y \circ d)$.
	Furthermore, bisimilarities are preserved by this functor.
\end{theorem}
It is then easy to assess that effect morphisms define functors that preserve bisimilarity, thanks to~\autoref{thm:morphisnattrans}.
Note however that functors induced by natural transformations do not reflect bisimilarities, in general.
In \cite{bartelsHierarchyProbabilisticSystem2004} further conditions are identified in order for functors to reflect bisimilarities.
\begin{restatable}{theorem}{alphainducef}
	Let $\alpha : F \Rightarrow G$ be a natural transformation between two functors $F$ and $G : \setcat \to \setcat$.
	If all components $\alpha_X$ of $\alpha$ are injective then the induced functor $\alpha \circ \blank$ reflects the kernel bisimilarity.
\end{restatable}

\subsection{The Quantum Case}


Hereafter, we apply the results above to the specific case of probability and quantum effect distributions and their monads.
Recall that, for brevity, we write $Q_C$ for $\emon[{\ef[\hilbert_C]}]$.
Note that coalgebras $(X, c: X \to \emon X)$ are essentially a generalization of Markov Chains, coinciding with the usual definition but with partial distributions if $\ealg = [0,1]$.
$Q_C$-coalgebras can be seen instead as a quantum version of Markov Chains.

We start by noticing that kernel and AM-bisimilarities coincide in the probabilistic case, thanks to~\autoref{thm:amkequality} and since $[0,1]$ is decomposable (as demonstrated implicitly in \cite{mossCoalgebraicLogic1999}).
On the contrary, the equality does not hold in the quantum case.
Thanks to~\autoref{thm:amkequality}, it suffices to show the following.
%
\begin{restatable}{proposition}{efnotdec}\label{thm:efnotdec}
	If the dimension of $\hilbert$ is greater or equal than $2$, then $\ef[\hilbert]$ is not decomposable.
\end{restatable}

A simple example of a case in which kernel and AM-bisimilarities disagree follows.
\begin{example}\label{ex:coalgebre}
	Take $X \xrightarrow{c} Q_C$, with $X = \{x_1,x_2,x_3,x_4\}$ and $c$ defined as below. 
	\begin{gather*}
	c(x_1) = \ketbra{0} \bullet x_3 + \ketbra{1} \bullet x_4 \qquad 
	c(x_2)  = \ketbra{+} \bullet x_3 + \ketbra{-} \bullet x_4 \qquad c(x_3) = c(x_4) = \mathbb{I} \bullet x_3
	\end{gather*}
	We cannot build an $AM$-bisimulation such that $x_1 \sim_{AM} x_2$.
	However, if we take $Z = \{z_1, z_2\}$ and $m$ such that $m(x_1) = m(x_2) = z_1$, $m(x_3) = m(x_4) = z_2$, then
	$(X,c) \xrightarrow{m} (Z,d) \xleftarrow{m} (X,c)$ is a cocongruence for $d$ defined as $d(z_1) = d(z_2) = \I_2 \bullet z_2$, i.e. $x_1 \sim_{k} x_2$. 
\end{example} 

Since kernel and AM-bisimilarities do not agree even on simple cases, we look for some correctness criteria to help us decide which one to use for comparing quantum systems.
For defining them, it is useful to recall the
effect homomorphism $m_\rho (L) = tr(L\rho)$ that for each $\rho \in \dm[\hilbert_C]$ maps a quantum effect to a probability, and the induced natural transformation from $Q_C$ to $\emon[{[0,1]}]$ by~\autoref{thm:morphisnattrans}.
Thanks to~\autoref{thm:nattransisfunc}, a functor $ \qtop\ $ is defined for each $\rho \in \dm[\hilbert_C]$ that maps a $Q_C$-coalgebra in the $\emon[{[0,1]}]$-coalgebra characterizing the behaviour of the Quantum Markov Chain when the input state is $\rho$.
Indeed, The functor $ \qtop\ $ updates the weights of a Quantum Markov Chain by computing the Born rule.
\begin{example}
Take the $Q_C$-coalgebra $(X,c)$ from~\autoref{ex:coalgebre}, then the $\emon[{[0,1]}]$-coalgebra, $\qtop[\ketbra{0}](X, c)$ is $(X, c')$ with $c'(x_1) = 1 \bullet x_3$, $c'(x_2)  = \frac{1}{2} \bullet x_3 + \frac{1}{2} \bullet x_4$ and $c'(x_3)  = c'(x_4)  = 1 \bullet x_3$.
\end{example}
Since their probabilistic behaviour is the only observable property of quantum systems,
our (first) correctness principle for a bisimilarity 
is that two elements of the $Q_C$-coalgebras $(X,c)$ and $(Y,d)$ shall be bisimilar if and only if they are indistinguishable in the probabilistic systems obtained by instantiating their quantum input, i.e.~in the $\emon[{[0,1]}]$-coalgebras $\qtop(X,c)$ and $\qtop(Y,d)$ for any $\rho \in \dm[\hilb_{C}]$.
\begin{desiderata}\label{des:uno}
Let $C$ be a collection of quantum systems, and let $(X, c)$ and $(Y, d)$ be $Q_C$-coalgebras with $x \in X$ and $y \in Y$.
Then, the two following conditions must coincide
\begin{itemize}
\item $x$ and $y$ are bisimilar in $(X, c)$ and $(Y, d)$
\item $x$ and $y$ are bisimilar in $ \qtop(X, c)$ and $ \qtop(Y, d)$ for all $\rho$
\end{itemize}
\end{desiderata}
By~\autoref{thm:nattransisfunc}, all the functors $\qtop\ $ preserve both bisimilarities, hence one side of the implication holds for both of them.
However,\autoref{thm:amkequality} and~\autoref{thm:efnotdec} imply that AM-bisimilarity in $\emon[{[0,1]}]$ is not reflected by $\qtop$ : it does not satisfy the other direction of~\autoref{des:uno}.
In particular, AM-bisimilarity fails to relate processes that are probabilistically indistinguishable when instantiated with all $\rho$. 
\begin{example}
	Take $(X,c)$ from~\autoref{ex:coalgebre}, and recall that $x_1 \sim_{k} x_2$, while $x_1 \not\sim_{AM} x_2$.
	From the discussion above we know that the former implies that $x_1 \sim_k x_2$ in $ \qtop(X, c)$ for all $\rho$.
	Moreover, the latter confirms that AM-bisimilarity does not meet our desiderata.
\end{example}

Regarding the other direction,
we obtain that kernel bisimilarity fully satisfies~\autoref{des:uno} when the quantum effects used in $(X,c)$ and $(Y,d)$ are contained in a finite effect algebra $\mathbb{L} \subsetneq \ef[\hilbert]$.
To prove this result, we build on the following technical lemma, restating our Lemma 4 of~\cite{lics}.
It ensures that, when $\mathbb{L}$ is finite, a single density operator exists that distinguishes every pair of distinct effects in $\mathbb{L}$.
\begin{lemma}\label{thm:injmorph}
	Let $\mathbb{L}$ be a finite set of quantum effects in $\ef[\hilbert]$, then a density operator $\widehat{\rho} \in \dm[\hilbert]$ exists such that for each $L, L' \in \mathbb{L}$, $tr(L \widehat{\rho}) = tr(L' \widehat{\rho})$ if and only if $L = L'$.
\end{lemma}

This means that among all the effect algebra homomorphisms $m_\rho$, there is at least an injective one $m_{\widehat{\rho}}$, 
which yields a natural transformation with injective components by~\autoref{thm:morphisnattrans}.
Thus, by \autoref{thm:nattransisfunc} the resulting functor $ \qtop[\widehat{\rho}]\ $ reflects kernel bisimilarity.

\begin{restatable}{theorem}{quantumkernelreflected}
	Let $(X, c)$ and $(Y, d)$ be two $\emon[\mathbb{L}]$-coalgebras, and let $x \in X$ and $y \in Y$.
	If $x$ and $y$ are kernel bisimilar in $ \qtop(X, c)$ and $ \qtop(Y, d)$ for all $\rho$, then $x$ and $y$ are kernel bisimilar in $(X, c)$ and $(Y, d)$.
\end{restatable}

As previously stated, this result is based on~\autoref{thm:injmorph}, which is proved only on finite effect algebras.
We expect that this result can be extended at least to the countably infinite case.
However, the current model is still capable of representing and comparing the behaviour of systems performing a finite set of different destructive measurements, which is often the case in quantum protocols~\cite{nurhadiQuantumKeyDistribution2018, bennettQuantumCryptographyPublic2014, longQuantumSecureDirect2007, gaoQuantumPrivateQuery2019} and in the quantum process calculi literature~\cite{ceragioliQuantumBisimilarityBarbs2024,kubotaApplicationProcessCalculus2012,fengAutomaticVerificationQuantum2015a,dengBisimulationsProbabilisticQuantum2018}.

We investigate now a different correctness principle, called locally parameterized probabilistic bisimilarity (lpp) in~\cite{lics}.
In the previous approach, two quantum systems are equated if they behave the same when given any possible quantum state.
Nevertheless, the input is given globally for the whole Markov Chain.
This means we assume that an adversary trying to disprove equivalence can only choose the state once at the beginning.
In lpp instead, the adversary can give a different quantum state at each step of the two Markov Chains.
In the following, we model this feature coalgebraically.
Recall that $L \mapsto \lambda \rho. tr(L \rho): {\ef[\hilbert]} \to \mathbf{Conv}(\dm[\hilbert], [0,1])$ is an isomorphism from quantum effects to convex parameterized probabilities.
Indeed, $\mathbf{Conv}(\dm[\hilbert_C], [0,1])$ is itself an effect algebra, that we name $f_C$ for convenience, and $L \mapsto \lambda \rho. tr(L \rho)$ is an effect algebra isomorphism.
We can therefore define $\emon[f_C]$ and a natural transformation $\alpha_{\mathit{lpp}}: Q_C \to \emon[f_C]$ transforming a quantum distribution over $X$ into a distribution that associates each element $x \in X$ with a function that given a quantum input returns a probability.
%
Intuitively, while $ \qtop\ $ applies a unique initial input density operator that transforms a Quantum Markov Chain into a probabilistic system,
$\alpha_{lpp} \circ \blank$ allow us to change $\rho$ at every step of our bisimilarity.
\begin{desiderata}\label{des:due}
Let $(X, c)$ and $(Y, d)$ be $Q_C$-coalgebras, and let $x \in X$ and $y \in Y$.
Then, bisimilarity of $x$ and $y$ in $(X, c)$ and $(Y, d)$, i.e. when effects are used as weights, must hold if and only if $x$ and $y$ are bisimilar in $(\alpha_{\mathit{lpp}} \circ \blank) (X, c)$ and $(\alpha_{\mathit{lpp}} \circ \blank) (Y, d)$, i.e. when convex functions are used as weights in place of their isomorphic effects.
\end{desiderata}
Kernel bisimilarity over $Q_C$-coalgebras exactly coincides with this notion of locally parameterized probabilistic bisimilarity.
\begin{restatable}{theorem}{ismorphism}
	Let $(X, c)$ and $(Y, d)$ be $Q_C$-coalgebras, and let $x \in X$ and $y \in Y$.
	Then
	$x \sim_k y$ in $(X, c)$ and $(Y, d)$ holds if and only if $x \sim_k y$ holds in $(\alpha_{\mathit{lpp}} \circ \blank) (X, c)$ and $(\alpha_{\mathit{lpp}} \circ \blank) (Y, d)$.
\end{restatable}
An interesting corollary is that, when coalgebras are built over a finite quantum effect algebra $\mathbb{L}$, the stronger adversary that can choose a different quantum input at each step is not capable of discriminating more than the weaker one that only chose once at the beginning.

\section{Effect Labelled Transition Systems}\label{sec:operators}

We can now define a generalization of probabilistic transition systems to effect distributions, for which we extend the properties of the previous section and introduce both a generic parallel composition and a specifically quantum operator of partial evaluation.
\begin{definition}[$\ealg$-Labelled Transition System]
	Given an effect algebra $\ealg$, a \emph{$\ealg$-labelled transition system} ($\ealg$LTS) is a coalgebra $X \xrightarrow{c} \mathcal{P}(\emon[\ealg]X)^L$, where $X$ is a set of states, $L$ is a fixed set of labels, $\powset$ is the finitary powerset endofunctor and $\emon$ is the $\ealg$-distribution endofunctor on $\setcat$.
\end{definition}
As is typical for process calculi, we will assume that there exist a special symbol $\tau \in L$ and an involutive unary operation $\overline{\blank}$ for all labels different from $\tau$. Intuitively, $\tau$ represents a silent, invisible action, and whenever $\mu$ represents a visible action (e.g. outputting a value), $\overline{\mu}$ represents its "dual" action (e.g. receiving that value).
Moreover, we will write $x \xrightarrow{\mu} \Delta$ for $\Delta \in c(x)(\mu)$.

In the following, we deal with $\emon[{[0,1]}]$LTSs and $\qmon[C]$LTSs (called pLTSs and qLTSs hereafter), i.e. LTSs over distributions of probabilities and of quantum effects in some $\hilbert_{C}$ respectively.
All the results on $\emon$-coalgebras of the previous section 
can be easily extended to more complex coalgebras through \emph{whiskering}.
For example, we can whisker the natural transformation $\downarrow_\rho : \qmon[C] \to \emon[{[0,1]}]$
and the finite powerset functor $\mathcal{P}$, obtaining a natural transformation $\powset\!\!\qtop\ : \powset\qmon[C] \to \powset\emon[{[0,1]}]$ and recovering the previous results for systems that transition non-deterministically to effect distributions.
\begin{corollary}
Let $(X, c)$ and $(Y,d)$ be qLTSs, with $x \in X$ and $y \in Y$.
	\begin{itemize}
		\item $x \sim_{AM} y \implies x \sim_{k} y$ in general for qLTSs;
		\item $x \sim_{k} y \centernot\implies x \sim_{AM} y$ if the dimension of the Hilbert space is at least two;
		\item for each $\rho$ over the adequate Hilbert space, there is a functor transforming $c$ and $d$ into the pLTSs $c_\rho$ and $d_\rho$ representing the probabilistic behaviour of the system when the quantum state is $\rho$;
		\item if $x$ and $y$ are kernel bisimilar in $c$ and $d$, then they are bisimilar also in $c_\rho$ and $d_\rho$ for every $\rho$;
		\item if $c$ and $d$ use a finite effect algebra of weights, and $\forall \rho. x \sim_{k} y$ in $c_\rho$ and $d_\rho$, then $x \sim_{k} y$ in $c$ and $d$;
		\item $x$ and $y$ are kernel bisimilar in $c$ and $d$ if and only if they are bisimilar in the $\emon[f_C]$LTSs where effects are considered as convex functions from density operators to probabilities, i.e. when the adversary can choose a different density operator at each step of the bisimilarity.
	\end{itemize}
\end{corollary}

In the rest of the section we describe some operators on $\ealg$LTSs.
When dealing with coalgebraic LTSs, it is typical to define operators like nondeterministic sum, restriction or parallel composition as acting on \textit{processes}, i.e. states of the final coalgebra of the functor under consideration~\cite{Jacobs2016}.
We instead define operators directly between coalgebras as in~\cite{ruttenUniversalCoalgebra2000}, allowing us to compose coalgebras of different functors.
We postpone to future work how these ``macroscopic'' operators relates to the ``microscopic'' ones in the final coalgebra, following the approach of~\cite{Hasuo11}.
We present two operators: the first generalizes parallel composition to all $\ealg$LTSs, the other is specific to qLTSs and specifies how to perform partial evaluation of the input quantum state.
We first introduce the notion of \emph{extensible} graded monad.


In $\cite{katsumata2014effect}$, the author proposes a monad graded on a preordered monoid, intended to model the composition of computational side effects and their scope.
We replicate this notion in our PCM setting, employing the fact that each PCM automatically carries a preorder structure.

\begin{definition}[Extensible Graded Monad]
	An $M$-graded monad is \emph{extensible} if for all $m \preceq n \in M$, a natural transformation $\xi_{m \preceq n}: T_m \to T_n$ called \emph{extension} exists, with
	$ \xi_{m \preceq m} = \mathit{Id}$ and $\xi_{n \preceq o} \circ \xi_{m \preceq n} = \xi_{m \preceq o}
	$.
\end{definition}

Extensible monads give us a way to canonically extend a computation to a greater grade, which we will need to define  parallel composition.  All non graded monad are trivially extensible  with the identity transformation.
We provide an extension for monads built on
effect monoids graded on an effect algebra.

\begin{restatable}{theorem}{effectgraded}
	Suppose $\{\ealg_m\}$ is an $M$-graded effect monoid and $M$ an effect algebra, with $-$ its induced difference operator. There is an extensible graded monad $\{\emon[\ealg_m]\}$ with unit $\eta: Id \to \emon[\ealg_0]$ and multiplication $\mu_{m,n}: \emon[\ealg_m]\emon[\ealg_n] \to \emon[\ealg_{m+n}]$ defined as in \autoref{thm:gmofgem} and extension $\xi_{m \preceq n}$
	\[
		\xi_{m \preceq n} (\sum\nolimits_i e_i \bullet x_i) = \sum\nolimits_i \nabla_{m,n-m}(e_i, 1_{\ealg_{n-m}}) \bullet x_i.
	\]
\end{restatable}
Indeed, $\qmon$ is extensible, since the partial monoid of quantum systems is an effect algebra.
\begin{restatable}{lemma}{extqmon}
	$\mathcal{S} = \langle \mathcal{P}(Sys), \emptyset, \uplus, Sys \setminus \blank \rangle$ is an effect algebra, with $\setminus$ the usual set difference.
\end{restatable}
Intuitively, when dealing with a computation consuming the set of qubits $C$, we can always extend it to a computation on $C' \supseteq C$ by performing the identity measurement on the extra qubits.

We now focus on the parallel composition of $\ealg$LTSs, in which two subsystems can move on their own or exchange messages.
In the quantum case, we will be able to compose qLTSs only when their grades are summable, ensuring that they do not perform measurements on the same quantum systems.

In the process-calculi literature, there are different notions of parallel composition, corresponding to different “synchronization styles”: à la CCS, CSP or ACP.
Moreover,
different extensions have been considered
from the original non-deterministic setting to the probabilistic \cite{hennessyExploringProbabilisticBisimulations2012,bartelsHierarchyProbabilisticSystem2004} and quantum \cite{lics} one.
We introduce a \textit{generic parallel operator}, which is parametric both with respect to the synchronization style and to the commutative graded monad chosen for the weights ($Id$ for LTSs,
$\emon[{[0,1]}]$ for pLTSs, $\qmon[C]$ for qLTSs).
Our parallel operator is defined in two steps: the first specifies how to compose the "non-deterministic structure" of the systems (the $\powset^L$ functor), the second how to combine a couple of distributions into a distribution of couples.
For the first step, we focus on a CCS-style interleaving composition, formalized as a binary \textit{synchronization operator} $\blank | \blank$ on transition functions.
\begin{definition}[CCS-Style Synchronization]
	Consider the extensible graded monad $\{T_m\}$.
	Given two coalgebras $X \xrightarrow{c} \PL{T_m X}$ and $Y \xrightarrow{d} \PL{T_n Y}$ we define $c | d : X \times Y  \to \PL{T_m X \times T_n Y}$, the "CCS-style" synchronization of $c$ and $d$, as
	\[
		\infer{s \parallel t \xrightarrow{\mu} \langle \Delta,\ \xi_{0 \preceq n} (\eta(t)) \rangle}{s \xrightarrow{\mu} \Delta}
		\qquad
		\infer{s \parallel t \xrightarrow{\mu} \langle \xi_{0 \preceq m} (\eta(s)) ,\ \Theta \rangle}{t \xrightarrow{\mu} \Theta}
		\qquad
		\infer{s \parallel t \xrightarrow{\tau} \langle \Delta,\ \Theta\rangle}
			{s \xrightarrow{\mu} \Delta & t \xrightarrow{\mu} \Theta}
	\]
	where $\times$ is the cartesian product on $\setcat$, $s \parallel t$ is an element of $X \times Y$, $\Delta$ (resp. $\Theta$) is an element of $T_m X$ (resp. $T_n Y$), and $s \xrightarrow{\mu} \Delta$ is the usual SOS-style notation for $\Delta \in c(s)(\mu)$.
\end{definition}

To model CCS-style synchronization we represent the “idle behaviour” of a state $s$ by taking the monad unit and extending it to the desired grade, i.e. $\xi_{0 \preceq m} (\eta(s))$.
In pLTSs this coincides with the unit $1 \bullet s$.
In qLTSs we get the point distribution "scaled" for the given dimension $\mathbb{I}_{d} \bullet s$.
For pure non-deterministic systems, where $T$ is the identity monad, we get the usual parallel composition of CCS.

For the second step, we define a parallel operator $\blank\parallel\blank$ on {$\ealg$LTSs}, in the form of a functor $\coalgcat[\PL{T_m}] \times \coalgcat[\PL{T_n}] \to \coalgcat[\PL{T_{m+n}}]$.
The definition is 
parametric with respect to a natural transformation $\alpha : T_mX \times T_nY \to T_{m+n}(X \times Y)$ which specifies how to combine two distributions.
For classical, probabilistic and quantum systems, we have commutative monads, which have a canonical transformation $\alpha$.
\begin{restatable}{theorem}{parallelfunctor}
	Given a transformation $\alpha : T_mX \times T_nY \to T_{m+n}(X \times Y)$, we can construct a functor $\parallel: \
		\coalgcat[\PL{T_m}] \times \coalgcat[\PL{T_n}] \to \coalgcat[\PL{T_{m+n}}]$ defined by
	\[ (X, c) \parallel (Y, d) = (X \times Y, \PL{\alpha} \circ (c | d)) \qquad f \parallel g = f \times g
	\]
\end{restatable}

We now consider a special operator acting over qLTS, i.e. the \textit{partial evaluation}.
\begin{definition}[Partial Evaluation of Quantum Effects]
	Consider a quantum effect $L \in \ef[\hilb_{C}]$.
	For all quantum state $\rho \in \dm[\hilb_{C'}]$ with $C' \subseteq C$, we define the \textit{partial evaluation} of $L$ with state $\rho$ as $tr_{C'}(L(\rho \boxtimes \mathbb{I}))$,
	where $\mathbb{I}$ is the identity operator on $\hilb_{C \setminus C'}$ and $\boxtimes$ is the product of \autoref{thm:efMonQuEf}.
\end{definition}
Our previously defined total evaluation $\downarrow_\rho$ is of course a specific case of partial evaluation, when $C' = C$.
For all $\rho$, partial evaluation $tr_{C'}(\blank(\rho \boxtimes \mathbb{I}))$ is an effect morphism, and thus yields a functor from $\coalgcat[{\qmon[C]}]$ to $\coalgcat[{\qmon[C \setminus C']}]$, and can be extended also to qLTSs via whiskering.
\begin{definition}[Partial Evaluation of qLTS]
		Let $C$, $C'$ be collections of quantum systems such that $C' \subseteq C$.
		For each $\rho \in \dm[\hilb_{C'}]$ we define the \textit{partial evaluation} of $\coalgcat[{\PL{\qmon[C]}}]$ with input $\rho$ as the functor
		$\peval{}{\rho} : \coalgcat[{\PL{\qmon[C]}}] \to \coalgcat[{\PL{\qmon[{C \setminus C'}]}}]$ induced by the effect morphism $L \mapsto tr_{C'}(L(\rho \boxtimes \mathbb{I}))$.
\end{definition}

Thanks to functoriality, parallel composition and partial evaluation preserve bisimilarity.
\begin{restatable}{theorem}{congruence}
	If $s \sim_k t$ then, for all $\rho$, $\peval{s}{\rho} \sim \peval{t}{\rho}$.
	Moreover, if $s' \sim t'$ then $s \parallel s' \sim t \parallel t'$.
\end{restatable}

\section{Conclusions}\label{sec:conclude}

We have characterized distributions with weights from a generic effect algebra, subsuming probability and quantum effect distributions.
We have introduced monads graded on a partial commutative monoid (PCM) that allow us to grade quantum distributions over their resources, which must be treated linearly, as prescribed by the no-cloning theorem.
We have studied effect weighted Markov Chains and labelled transition systems ($\ealg$LTS) in a coalgebraic framework, extending previous results about kernel and Aczel-Mendler bisimilarities.
We have applied our findings to the quantum setting, proving that each quantum state $\rho$ defines a functor from qLTSs to pLTSs that “instantiates” a quantum process to the probabilistic behaviour it exhibits when the quantum state is $\rho$.
We have compared the two notions of bisimilarity in hand with the desired properties of a behavioural equivalence for quantum systems, proving that the kernel bisimilarity is the only one of the two that captures theirs observable features, i.e.~the induced probabilistic behaviour.
Finally, we have defined parallel composition of $\ealg$LTSs in a functorial, compositional way, and a partial evaluation operator that given a qLTS instantiates some of its input qubits, paving the way for an $\ealg$LTS semantics of quantum process calculi.

\paragraph*{Future work}
We will study the final coalgebra of our functors, investigating the definition of graded GSOS operators for our graded $\ealg$LTSs, and their relation with the functorial ones defined in this work. 
Our framework allows us to compose effects in parallel via the tensor product.
We will consider also sequential composition, extending our graded approach to the case of superoperators, which model more general quantum operations.
We expect such an extension to come naturally and to preserve our constructions and results.
Finally, we will explore the reflection of the kernel bisimilarity in broader cases, i.e.~when the quantum effect algebra is finitely generated, but possibly infinite, and when it is numerable in general.
Recent work~\cite{effectfulstreams} defines bisimulation for Mealy machines with generic effects.
We intend to look at the relationship between their effectful bisimilarity and our notion of kernel bisimilarity.

\bibliographystyle{eptcs}
\bibliography{quantumVerification}

\clearpage

\appendix
\section{Proof Sketches}

\morphisnattrans*
\begin{proof}[Proof sketch]
	Naturality of \(m \circ \blank\) follows from the fact that \(m\) preserves the (partial) sum and that distributions have finite support.
	If \(m\) is injective, the components of \(m \circ \blank\) are also injective because if two different \( \Delta, \Theta \in \emon[{\ealg[E]}](X)\) differ on \(x \in X\), then \(m \circ \Delta\) and \(m \circ \Theta\) will differ on the same \(x\).
\end{proof}

The evolution of density operators is given as a \emph{trace preserving
	superoperator} $\mathcal{E}: \dm[\hilbert_A] \rightarrow \dm[\hilbert_B]$, a function defined by its \emph{Kraus operator sum decomposition} $\{E_i\}_i$ for
a finite set of indexes $i = 1, \ldots,
	n \times m$, satisfying that
$E_i \in \mathbb{C}^{m \times n}, \E({\rho}) = \sum_i E_i\rho E_i^\dag$ and $\sum_i E_i^\dag E_i = \I_n$, where $n$ and $m$ are the dimension of Hilbert space $\hilbert_A$ and $\hilbert_B$ respectively.
The tensor product of density operators $\rho \otimes \sigma$ is defined as their Kronecker product, and of superoperators $\mathcal{E} \otimes
	\mathcal{F}$ as the superoperator having Kraus decomposition $\{ E_i \otimes
	F_j \}_{i,j}$ with $\{ E_i \}_i$ and $\{ F_j \}_j$ Kraus decompositions of
$\mathcal{E}$ and $\mathcal{F}$.

A morphism between effects is the dual of a superoperator.
Let $\mathcal{E}(\rho) = \sum_i E_i \rho E_i^\dag$ be a superoperator.
Its dual is the superoperator $\mathcal{E}^\dag(L) = \sum_i E_i^\dag L E_i$.

\efMonQuEf*
\begin{proof}[Proof sketch]
	Recall that $Sort_{C,D}$ is the transformation on quantum effects (i.e.\ a dual superoperator) which applies a unitary permutation matrix from $\hilbert_C \otimes \hilbert_D$ to $\hilbert_{C \uplus D}$.
	Unitality follows from the definition of Kronecker product and from $Sort_{C, \emptyset}$ being the identity transformation.
	For commutativity, we check that $Swap_{C,D}$, the superoperator that permutes $\hilbert_C \otimes \hilbert_D$ in $\hilbert_D \otimes \hilbert_C$, respects sorting:
	\[
		Sort_{D, C} (L_2 \otimes L_1) = Sort_{D, C} (Swap_{C, D} (L_1 \otimes L_2)) = (Sort_{D, C} \circ Swap_{C, D}) (L_1 \otimes L_2) = Sort_{C, D} (L_1 \otimes L_2).
	\]
	%
	%
	%
	For associativity, we have that
	\[
		Sort_{C, D \uplus E} (L_1 \otimes Sort_{D,  E} (L_2 \otimes L_3)) = (Sort_{C, D \uplus E} \circ (Id \otimes Sort_{D,  E})) (L_1 \otimes L_2 \otimes L_3) = Sort_{C, D, E} (L_1\otimes L_2 \otimes L_3)
	\]
	and similarly for $Sort_{C \uplus D, E} (Sort_{C, D} (L_1 \otimes L_2) \otimes L_3)$, where $Id$ is the identity superoperator, and $Sort_{C, D, E}$ is the superoperator going from $\hilbert_C \otimes \hilbert_D \otimes \hilbert_E$ to $\hilbert_{C \uplus D \uplus E}$.
\end{proof}

\gmofgem*
\begin{proof}[Proof sketch]
	The proof proceeds as in the total case: the unit and multiplication of the monad are defined with the graded monoid structure of the weights, and the associativity and unitality conditions follow from the graded monoid ones.
\end{proof}

\commute*
\begin{proof}[Proof sketch]
	As any $\setcat$-endofunctor, $\emon[\ealg_m]$ has a canonical left-strength $\sigma_{m, X, Y}: X \times \emon[\ealg_m]Y \to \emon[\ealg_m](X \times Y)$ given by
	\[(x, (\sum_i e_i \bullet y_i)) \mapsto \sum_i e_i \bullet x,y_i\]
	and a corresponding right strength given by the monoidal structure of $\setcat$. Both transformations yield indeed a strong graded monad.
	It's easy to see that the morphism $\mu_{m,n} \circ \emon[\ealg_m]\sigma_n \circ \tau_m$ brings $(\sum_i e_i \bullet x_i,  \sum_j e_j' \bullet y_j)$
	into $\sum_{i, j} \nabla_{n, m}(e'_j, e_i) \bullet (x, y)$, while $\mu_{n,m} \circ \emon[\ealg_n]\tau_m \circ \sigma_n$ brings it to $\sum_{i, j} \nabla_{m, n}(e_i, e'_j) \bullet (x, y)$. The two coincides whenever $\nabla$ is commutative.
\end{proof}

\amkequality*
\begin{proof}[Proof sketch]
	The first point has been demonstrated in \cite[Corollary 8]{bartelsHierarchyProbabilisticSystem2004}.

	The \emph{if} case of the second point follows from the composition of two previously demonstrated results.
	In \cite[Theorem 5.13]{ruttenUniversalCoalgebra2000}, the authors prove that if a commutative monoid is positive and decomposable then the functor $\emon$ weakly preserves weak pullbacks.
	Such demonstration is also applicable to partial commutative monoid, of which effect algebras are a special case.
	Again, \cite[Corollary 8]{bartelsHierarchyProbabilisticSystem2004} prove that if $\emon$ preserves weak pullback then $\sim_k$ implies $\sim_{AM}$.

	Finally, the \emph{only if} case is demonstrated through contrapositivity.
	Assume we have effects $a,b,c,d \in \ealg$ that are not decomposable, i.e.\ $a + b = c + d$ and there are no $e_{11},e_{12},e_{21},e_{22} \in \ealg$ such that
	$a = e_{11} + e_{12}$, $b = e_{21} + e_{22}$, $c = e_{11} + e_{21}$ and $d = e_{12} + e_{22}$. Let $s = a + b = c + d$.

	Consider three coalgebras:
	\begin{enumerate}
		\item $(X, c)$ with $X = \{x_1,x_2,x_3\}$ and $c(x_1) = \{x_2 \mapsto a, x_3 \mapsto b\}$.
		\item $(Y, d)$ with $Y = \{y_1,y_2,y_3\}$ and $d(y_1) = \{y_2 \mapsto c, y_3 \mapsto d\}$.
		\item $(Z, z)$ with $Z = \{z_1,z_2\}$ and $z(z_1) = \{z_2 \mapsto s\}$.
	\end{enumerate}
	It is straightforward to show that $X \xrightarrow{m_1} Z \xleftarrow{m_2} X$ with $m_1(x_1) = m_2(y_1) = z_1$ and $m_1(x_2) = m_1(x_3) = m_2(y_1) = m_2(y_2) = z_2$, $m_1(x_4) = m_2(x_4) = z_3$ is a cocongruence, therefore $x_1 \sim_k y_1$.
	By~\cite[Lemma 5.5]{ruttenUniversalCoalgebra2000}, a relation $R \subseteq X \times Y$ is an AM-bisimulation if and only if for every $(a, b) \in R$ there is a $|X|\times|Y|$ matrix, $(m_{x,y})$, with entries from $\ealg$ such that
	\begin{itemize}
		\item there are all but finitely many $m_{x,y} = 0$;
		\item if $m_{x,y} \neq 0$ then $(x, y) \in R$;
		\item $\forall x \in X \ldotp c(a)(x) = \sum_{y \in Y} m_{x,y}$;
		\item $\forall y \in Y \ldotp d(b)(y) = \sum_{x \in X} m_{x,y}$;
	\end{itemize}
	Let us consider the case for the pair $(x_1, y_1)$.
	The requirements can be represented by the following table.
	\begin{center}
		\begin{tabular}{c c c|l}
			$0$ & $0$           & $0$           & $0$ \\
			$0$ & $m_{x_2,y_2}$ & $m_{x_2,y_3}$ & $a$ \\
			$0$ & $m_{x_3,y_2}$ & $m_{x_3,y_3}$ & $b$ \\
			\midrule
			$0$ & $c$           & $d$           &
		\end{tabular}
	\end{center}
	However, by assumption of non-decomposability there are no $m_{x_2,y_2}$, $m_{x_2,y_3}$, $m_{x_3,y_2}$, $m_{x_3,y_3}$ that can satisfy such requirements, hence $x_1 \nsim_{AM} y_1$.
\end{proof}

\alphainducef*
\begin{proof}[Proof sketch]
	Let $X \xrightarrow{m_1} Z \xleftarrow{m_2} Y$ be a cocongruence between any two coalgebras $X \xrightarrow{\alpha_X \circ c} GX$ and $Y \xrightarrow{\alpha_Y \circ d} GY$.
	Let $R$ be the pullback in $\setcat$ of such cospan, witnessing the bisimilarity between the two $G$-coalgebras.
	\cite[Theorem 5]{bartelsHierarchyProbabilisticSystem2004} prove that the same cocongruence is also witness for the kernel bisimulation between the coalgebras $X \xrightarrow{c} FX$ and $Y \xrightarrow{d} FY$.
	Therefore, $R$ is also witness of the bisimilarity between the two $F$-coalgebras.
\end{proof}

\efnotdec*
\begin{proof}[Proof sketch]
	Take the following equality $\ketbra{0} + \ketbra{1} = \ketbra{+} + \ketbra{-}$.
	By \cite[Proposition 1.63]{heinosaariMathematicalLanguageQuantum2011}, if $R$ is a positive rank-$1$ operator, $T$ a positive operator and $T \sqsubseteq R$, then $T = pR$ for some $p \in [0,1]$.
	Recall that the partial order of quantum effects is the L\"owner order, and that the considered effects are positive rank-$1$ operator.
	Let $\ketbra{0} = e_{11} + e_{12}$ for some $e_{11},e_{12} \in \hilbert_C$.
	It must be that $e_{11} = p_{11}\ketbra{0}$ for some $p_{11} \in [0,1]$, and similarly $e_{12} = p_{12}\ketbra{0}$, with $p_{11} + p_{12} = 1$.
	But in order to satisfy the decomposability requirement, $\ketbra{+} = e_{11} + e_{21}$ for some $e_{21} \in \hilbert_C$.
	Then, $e_{11} = p'_{11}\ketbra{+}$ for some $p_{11}' \in [0,1]$.
	Thus, $p_{11} \ketbra{0} = p'_{11} \ketbra{+}$, which holds if and only if $p_{11} = p'_{11} = 0$.
	Assume without loss of generality that $e_{11} = 0$.
	The same line of reasoning can show that $p_{12}\ketbra{0} = p'_{12}\ketbra{-}$ holds if and only if $p_{12} = p'_{12} = 0$.
	Hence, both $e_{11}$ and $e_{12}$ must be $0$, thus it is impossible to decompose $\ketbra{0} + \ketbra{1} = \ketbra{+} + \ketbra{-}$. 
	This construction generalizes easily to dimensions grater than two.\qedhere
\end{proof}

\quantumkernelreflected*
\begin{proof}[Proof sketch]
	Assume $x$ and $y$ are kernel bisimilar in $\qtop(X, c)$ and $\qtop(Y, d)$ for all $\rho$.
	By~\autoref{thm:injmorph} a $\widehat{\rho}$ such that the effect algebra homomorphism $L \mapsto tr(L\widehat{\rho})$ is injective.
	Then, the components of the natural transformation induced by this homomorphism are injective~\autoref{thm:morphisnattrans}.
	By assumption, $x$ and $y$ are bisimilar also in $\qtop[\widehat{\rho}](X, c)$ and $\qtop[\widehat{\rho}](Y, d)$.
	The result then follows by~\autoref{thm:nattransisfunc}, as $\downarrow_{\widehat{\rho}}$ reflects kernel bisimilarity.
\end{proof}

\ismorphism*
\begin{proof}[Proof sketch]
	Assume $x$ and $y$ are kernel bisimilar in $(X, c)$ and $(Y, d)$, then the injective effect algebra homomorphism $L \mapsto \lambda \rho. tr(L \rho)$ defines a natural transformation $\alpha_{lpp}: Q_c \to \emon[f_C]$ with injective components~\autoref{thm:injmorph}.
	By~\autoref{thm:nattransisfunc}, $\alpha_{lpp} \circ \blank$ is a functor from $\coalgcat[Q_C]$ to $\coalgcat[f_C]$: it preserves both bisimilarities as every functor, moreover, it also reflects kernel bisimilarity since $\alpha_{lpp}: Q_c \to \emon[f_C]$ has injective components.
\end{proof}

\effectgraded*
\begin{proof}[Proof sketch]
	The desired conditions follow from the unitality and associativity of $\nabla$, and from the fact that $\nabla_{m,n} (1,1) = 1 $, since it is an effect bihomomorphism.
\end{proof}

\extqmon*
\begin{proof}[Proof sketch]
	$Sys \setminus C_1 = C_2$ is the unique element in $\mathcal{P}(Sys)$ such that $C_1 \uplus C_2 = Sys = Sys \setminus \emptyset$ by definition.
	Finally, by construction, $C \uplus Sys$ is defined only if $C \cap Sys = \emptyset$, i.e. $C$ must be $\emptyset$.
\end{proof}

\begin{restatable}{lemma}{transpar}\label{lem:transPar}
	Given two coalgebra homomorphisms $f: c \to c', g: d \to d'$, the synchronization operator $\blank | \blank$ makes the following diagram commute
	\[
		\begin{tikzcd}[ampersand replacement=\&]
			X \times Y \arrow[d, "c | d"] \arrow[rr, "f \times g"]                                 \&  \& X' \times Y' \arrow[d, "c' | d'"]     \\
			{\PL{T_m X \times T_n Y}} \arrow[rr, "{\PL{T_m f \times T_n g}}"'] \& \& {\PL{T_m X' \times T_n Y'}}
		\end{tikzcd}
	\]
\end{restatable}
\begin{proof}[Proof sketch]
	The commutativity of the above diagram correspond to verifying that
	\begin{itemize}
		\item For all $s \parallel t \xrightarrow{\mu} \langle \Delta, \Theta \rangle$, there exists $\Delta', \Theta'$ such that $f(s) \parallel g(t) \xrightarrow{\mu} \langle \Delta' , \Theta' \rangle$ and $\Delta' = f(\Delta), \Theta' = g(\Theta)$
		\item For all $f(s) \parallel g(t) \xrightarrow{\mu} \langle \Delta' , \Theta' \rangle$, there exists $\Delta, \Theta$ such that $s \parallel t \xrightarrow{\mu} \langle \Delta, \Theta \rangle$ and $\Delta' = f(\Delta), \Theta' = g(\Theta)$
	\end{itemize}
	which we can prove by cases on the definition on $\blank | \blank$, thanks to the hypothesis on $f$, which correspond to
	\begin{itemize}
		\item For all $s \xrightarrow{\mu} \Delta$, there exists $\Delta'$ such that $f(s) \xrightarrow{\mu} \Delta' $ and $\Delta' = f(\Delta)$
		\item For all $f(s) \xrightarrow{\mu} \Delta' $, there exists $\Delta$ such that $s \xrightarrow{\mu} \Delta $ and $\Delta' = f(\Delta)$
	\end{itemize}
	and similarly for $g$.
\end{proof}

\parallelfunctor*
\begin{proof}[Proof sketch]
	Checking that $\parallel$ is a functor amounts to checking the commutativity of the following diagram
	\[
		\begin{tikzcd}
			X \times Y \arrow[d, "c | d"] \arrow[rr, "f \times g"]                                                     &  & X' \times Y' \arrow[d, "c' | d'"]                         \\
			{\PL{T_m X \times T_n Y}} \arrow[rr, "{\PL{T_m f \times T_n g}}"'] \arrow[d, "\PL{\alpha}"] &  & {\PL{T_m X' \times T_n Y'}} \arrow[d, "\PL{\alpha}"] \\
			{\PL{T_{m+n} (X \times Y)}} \arrow[rr, "{\PL{T_{m+n} (f \times g)}}"']                 &  & {\PL{T_{m+n} (X' \times Y')}}
		\end{tikzcd}
	\]
	which holds since the bottom square commutes by naturality of $\alpha$, and the top squares commute by \autoref{lem:transPar}.
\end{proof}

\congruence*
\begin{proof}[Proof sketch]
	For the first, notice that $L \mapsto \peval{L}{\rho}$ is an effect morphism thanks to the linearity of the partial trace.
	Thus, it yields a natural transformation and a functor by~\autoref{thm:morphisnattrans} and \autoref{thm:nattransisfunc}, and functors preserve bisimilarity.
	For the second, thanks to functoriality, we know that $\parallel$ maps cospans in cospans:
	\begin{equation*}
		\left(
		\begin{tikzcd}
			{(X, c)} \arrow[d, "f"] & {(Y, d)} \arrow[ld, "g"] \\
			{(Z, z)}
		\end{tikzcd}
		\quad\bigparallel\quad
		\begin{tikzcd}
			{(X', c')} \arrow[d, "f'"] & {(Y', d')} \arrow[ld, "g'"] \\
			{(Z', z')}
		\end{tikzcd}
		\right)
		=
		\begin{tikzcd}
			{(X, c) \parallel (X', c') } \arrow[d, "f \times g"] & {(Y, d) \parallel (Y', d')} \arrow[ld, "f' \times g'"] \\
			{(Z, z) \parallel (Z', z')}
		\end{tikzcd}
	\end{equation*}
	Then, it is easy to see that the $\setcat$-pullback of $X \times X' \rightarrow  Z \times Z' \leftarrow Y \times Y'$ contains all and only the couples $s\parallel s', t \parallel t'$ such that $f(s) = g(t)$ and $f(s') = g(t')$. In other words, the cartesian product of two bisimulations is a bisimulation.
\end{proof}

\end{document}